\documentclass[10pt]{article}
\pdfoutput=1

\usepackage{fullpage} 

\usepackage[sumlimits]{amsmath}
\usepackage{amsfonts, amssymb, amsthm}

\usepackage{latexsym}
\usepackage{url}
\usepackage{xcolor}
\definecolor{DarkGray}{rgb}{0.1,0.1,0.5}
\usepackage[colorlinks=true,breaklinks, linkcolor= DarkGray,citecolor= DarkGray,urlcolor= DarkGray]{hyperref}
\usepackage{graphicx}
\usepackage{color}
\usepackage[tight, TABBOTCAP]{subfigure}
\usepackage{array}

\usepackage[tracking=smallcaps]{microtype}	

\def\place #1#2#3{\mspace{#2}\makebox[0pt]{\raisebox{#3}{#1}}\mspace{-#2}}	

\usepackage{todonotes}

\renewcommand{\todo}[1]{}

\newtheorem{thm}{Theorem}

\newtheorem{cor}[thm]{Corollary}
\newtheorem{prp}[thm]{Proposition}
\newtheorem{lem}[thm]{Lemma}
\newtheorem{definition}[thm]{Definition}
\newtheorem{clm}[thm]{Claim}
\theoremstyle{remark}

\newcommand{\refsec}[1]{\hyperref[sec:#1]{{Section~\ref*{sec:#1}}}}
\newcommand{\refalg}[1]{\hyperref[alg:#1]{{(Algorithm~\ref*{alg:#1})}}}
\newcommand{\reffig}[1]{\hyperref[fig:#1]{{Figure~\ref*{fig:#1}}}}
\newcommand{\refeqn}[1]{\hyperref[eqn:#1]{{(\ref*{eqn:#1})}}}
\newcommand{\reftbl}[1]{\hyperref[tbl:#1]{{Table~\ref*{tbl:#1}}}}
\newcommand{\refapp}[1]{\hyperref[app:#1]{{Appendix~\ref*{app:#1}}}}
\newcommand{\refdef}[1]{\hyperref[def:#1]{{Definition~\ref*{def:#1}}}}
\newcommand{\refthm}[1]{\hyperref[thm:#1]{{Theorem~\ref*{thm:#1}}}}
\newcommand{\reflem}[1]{\hyperref[lem:#1]{{Lemma~\ref*{lem:#1}}}}
\newcommand{\refcor}[1]{\hyperref[cor:#1]{{Corollary~\ref*{cor:#1}}}}
\newcommand{\refprp}[1]{\hyperref[prp:#1]{{Proposition~\ref*{prp:#1}}}}
\newcommand{\refclm}[1]{\hyperref[clm:#1]{{Claim~\ref*{clm:#1}}}}

\newcommand{\R}{\mathbf R}
\newcommand{\C}{\mathbf C}
\newcommand{\N}{\mathbf N}
\newcommand{\Z}{\mathbf Z}
\newcommand{\eps}{\varepsilon}

\newcommand{\cD}{{\cal D}}

\newcommand{\cH}{{\cal H}}
\newcommand{\tcH}{{\cal \tilde H}}

\newcommand{\cP}{{\cal P}}
\newcommand{\cC}{{\cal C}}

\newcommand{\Vf}{V_{\mathrm{free}}}
\newcommand{\wf}{w_{\mathrm{free}}}

\newcommand{\bra}[1]{{\langle#1|}}
\newcommand{\ket}[1]{{|#1\rangle}}
\newcommand{\braket}[2]{{\langle#1|#2\rangle}}
\newcommand{\ketbra}[2]{{\ket{#1}\!\bra{#2}}}

\newcommand{\abs}[1]{{\lvert #1\rvert}}	

\newcommand{\norm}[1]{{\| #1 \|}}

\newcommand{\pfstart}{\begin{proof}} 
\newcommand{\pfend}{\end{proof}}

\DeclareMathOperator{\spn}{span}

\DeclareMathOperator{\wsize}{wsize}

\def\adjoint{\dagger}
\newcommand{\binomial}[2]{\ensuremath{\left(\begin{smallmatrix}#1 \\ #2 \end{smallmatrix}\right)}}

\newcommand{\target}{\tau}

\title{
\vspace{-1.0cm}
Span programs and quantum algorithms \\ for $st$-connectivity and claw detection}
\author{Aleksandrs Belovs\thanks{Faculty of Computing, University of Latvia, Raina bulv. 19, Riga, LV-1586, Latvia. \texttt{stiboh@gmail.com}.}\and Ben W. Reichardt\thanks{Department of Electrical Engineering, University of Southern California.  \texttt{ben.reichardt@usc.edu}.}}
\date{}

\begin{document}
\maketitle

\vspace{-.2cm}
\begin{abstract}
We introduce a span program that decides $st$-connectivity, and generalize the span program to develop quantum algorithms for several graph problems.  First, we give an algorithm for $st$-connectivity that uses $O(n \sqrt d)$ quantum queries to the $n \times n$ adjacency matrix to decide if vertices~$s$ and~$t$ are connected, under the promise that they either are connected by a path of length at most~$d$, or are disconnected.  We also show that if $T$ is a path, a star with two subdivided legs, or a subdivision of a claw, its presence as a subgraph in the input graph~$G$ can be detected with~$O(n)$ quantum queries to the adjacency matrix.  Under the promise that~$G$ either contains~$T$ as a subgraph or does not contain~$T$ as a minor, we give $O(n)$-query quantum algorithms for detecting~$T$ either a triangle or a subdivision of a star.  All these algorithms can be implemented time efficiently and, except for the triangle-detection algorithm, in logarithmic space.  
One of the main techniques is to modify the $st$-connectivity span program to drop along the way ``breadcrumbs," which must be retrieved before the path from~$s$ is allowed to enter~$t$.  
\end{abstract}

\vspace{-.1cm}
\section{Introduction} \label{sec:introduction}

Span programs are a linear-algebraic way of specifying boolean functions~\cite{KarchmerWigderson93span}.  They are equivalent to quantum query algorithms; the least span program witness size for a boolean function is within a constant factor of the bounded-error quantum query complexity~\cite{Reichardt09spanprogram_arxivandfocs, Reichardt10advtight}.  To date, quantum algorithms have been developed based on span programs for formula evaluation~\cite{ReichardtSpalek08spanprogram, Reichardt09andorfaster, Reichardt09unbalancedformula}, matrix rank~\cite{Belovs11rank}, subgraph detection~\cite{Belovs11triangle, Zhu11learninggraph, LeeMagniezSantha11learninggraph}, and the $k$-distinctness problem under a certain promise~\cite{BelovsLee11learninggraph}.  

In this paper we give two new applications for span programs.  First, we present a new quantum algorithm for the $st$-connectivity problem, that uses exponentially less space and runs faster in many cases than the previous best algorithm.  Second, we give a quantum algorithm for detecting arbitrary fixed paths and claws in a graph.  All of our algorithms can be implemented time efficiently.

\vspace{-.1cm}
\paragraph{Quantum algorithm for deciding $st$-connectivity.}

In the (undirected) $st$-connectivity problem, we are given an undirected $n$-vertex graph $G$ with two selected vertices $s$ and~$t$.  $G$ is given by its adjacency matrix, i.e., the $n\times n$ symmetric matrix $(x_{ij})$, where~$x_{ij} = 1$ if the edge $(i, j)$ is present in the graph, and~$x_{ij} = 0$ otherwise.  The task is to determine whether there is a path from $s$ to~$t$ in $G$.  This problem is also known as {USTCON} or {UPATH}.  Classically, it can be solved in quadratic time by a variety of algorithms.  Its randomized query complexity is $\Theta(n^2)$.  With more time, it can be solved in logarithmic space~\cite{Reingold05stconn, AleliunasKarpLiptonLovaszRackoff79upathRL}.  

D{\"u}rr {\em et al.}\ have given a quantum algorithm for $st$-connectivity that makes $O(n^{3/2})$ queries to the adjacency matrix~\cite{DurrHeiligmanHoyerMhalla04graph}.  In fact, with an approach based on Bor\r{u}vka's algorithm~\cite{Boruvka26graph}, they solve a more general problem and find a minimum spanning tree in~$G$, i.e., a cycle-free edge set of maximal cardinality that has minimum total weight.  In particular, the algorithm outputs a list of the connected components of the graph.  The algorithm's time complexity is also $O(n^{3/2})$ up to logarithmic factors.  The algorithm works by executing a quantum subroutine that uses $O(\log n)$ qubits and requires coherently addressable access to $O(n \log n)$ classical bits, or quantum RAM~\cite{GiovannettiLloydMaccone07quantumRAM}.  This memory is changed classically between runs of the quantum subroutine.  

Our algorithm has the same time complexity as that of D\"urr {\em et al.}\ in the worst case, and has only logarithmic space complexity.  Moreover, the time complexity reduces to $\tilde O(n \sqrt d)$, if it is known that the shortest path between $s$ and $t$, if the one exists, has length at most $d$.  (The $\tilde O$ notation suppresses poly-logarithmic factors.)  Note, though, that our algorithm only detects the presence of a path, and does not output a path.  

The algorithm has a very simple form.  It works in the Hilbert space $\C^{\binomial{n}{2}}$, with one dimension per possible edge in the graph.  It alternates two reflections.  The first reflection applies a coherent query to the adjacency matrix input in order to add a phase of $-1$ to all edges not present in~$G$.  The second reflection is a reflection about all balanced flows from $s$ to~$t$ in the complete graph~$K_n$.  The second reflection is equivalent to reflecting about the constraints that for every vertex $v \notin \{s, t\}$, the net flow into~$v$ must be zero.  This is difficult to implement directly because constraints for different vertices do not commute with each other.  Our time-efficient procedure essentially works by inserting a new vertex in the middle of every edge.  Balanced flows in the new graph correspond exactly to balanced flows in the old graph, but the new graph is bipartite.  We reflect separately about the constraints in each half of the bipartition, and combine these reflections with a phase-estimation subroutine.

\paragraph{Subgraph containment graph properties.}

Using the $st$-connectivity algorithm as a subroutine and the color-coding technique from~\cite{AlonYusterZwick95colorcoding}, we can give an optimal algorithm for detecting the presence of a length-$k$ path in a graph~$G$ given by its adjacency matrix.  Assign to each vertex of~$G$ a label or ``color" chosen from $\{1, 2, \ldots, k+1\}$, independently and uniformly at random.  Discard the edges of~$G$ except those between vertices with consecutive colors.  Add two new vertices $s$ and $t$, and join $s$ to all vertices of color~$1$, and~$t$ to all vertices of color $k+1$.  If there is a path from~$s$ to~$t$ in the resulting graph~$H$, then $G$ contains a length-$k$ path.  Conversely, if $G$ contains a length-$k$ path, then with probability at least $2 (k+1)^{-k-1} = \Omega(1)$ the vertices of the path are colored consecutively, and hence $s$ and $t$ are connected in~$H$.  The algorithm's query complexity is $O(n \sqrt d) = O(n)$, using $d = k+3$.  The previous best quantum query algorithms for deciding if a graph contains a length-$k$ path use $\tilde O(n)$ queries for $k \leq 4$, $\tilde O(n^{3/2 - 1/(\lceil k/2\rceil-1)})$ queries for~$k \geq 9$, and certain intermediate polynomials for $5 \leq k \leq 8$~\cite{ChildsKothari10graphproperties}.  

Path detection is a special case of the problem of deciding whether~$G$ contains as a subgraph a certain fixed graph~$T$.  Algorithms applicable to general~$T$, with complexities depending on the number of vertices and their degrees in~$T$, have been given by~\cite{MagniezSanthaSzegedy05triangle, LeeMagniezSantha11learninggraph}.  A useful case is when $T$ is a triangle.  Quantum query algorithms for triangle-finding have improved from using $O(n^{3/2})$ queries, by Grover's algorithm, to $\tilde O(n^{1.3})$ queries~\cite{MagniezSanthaSzegedy05triangle}, to $O(n^{1.\overline{296}})$ queries~\cite{Belovs11triangle}.  Another case that has been studied is for $T$ a subdivision of a claw.  For detecting a $\{k_1, k_2, k_3\}$-claw, i.e., the star~$K_{1,3}$ with three legs subdivided into paths of lengths~$k_1$, $k_2$ and~$k_3$, Childs and Kothari have given an $\tilde O(n^{3/2-2/(k_1+k_2+k_3-1)})$-query algorithm if all $k_j$s are even, with a similar expression if any of them is odd~\cite{ChildsKothari10graphproperties}.  The best known lower bound for all these problems is just $\Omega(n)$ (see \refprp{optimal}).  

Subgraph-containment properties are a special case of {\em forbidden subgraph properties}, i.e., properties characterized by a finite set of forbidden subgraphs.  Another class of graph properties are {\em minor-closed} graph properties, i.e., properties that if satisfied by a graph $G$ are also satisfied by all minors of~$G$.  Natural examples include whether the graph is acyclic, and whether it can be embedded in some surface.  The properties of not containing a length-$k$ path or a $\{k_1, k_2, k_3\}$-claw are also minor closed.  Robertson and Seymour have famously shown that any minor-closed property can be described by a finite set of forbidden minors~\cite{RobertsonSeymour04wagner}.  They also have developed a cubic-time deterministic algorithm for solving any minor-closed graph property~\cite{RobertsonSeymour95minoralgorithm}.  Childs and Kothari have shown that the quantum query complexity of a minor-closed graph property is $\Theta(n^{3/2})$ unless the property is a forbidden subgraph property, in which case it is $o(n^{3/2})$~\cite{ChildsKothari10graphproperties}.  

We make further progress on characterizing the quantum query complexity of minor-closed forbidden subgraph properties.  In particular, we show that a minor-closed property can be solved by a quantum algorithm that uses $O(n)$ queries and $\tilde O(n)$ time if it is characterized by a {\em single} forbidden subgraph.  The graph is then necessarily a collection of disjoint paths and subdivided claws.  This query complexity is optimal.  The algorithm for these cases is a generalization of the $st$-connectivity algorithm.  It still checks connectivity in a certain graph built from~$G$, but also pairs some of the edges together so that if one edge in the pair is used then so must be the other.  Roughly, it is as though the algorithm drops breadcrumbs along the way that must be retrieved before the path is allowed to enter~$t$.  

For an example of the breadcrumb technique, consider the problem of deciding whether~$G$ contains a triangle.  We might attempt to solve this problem by first randomly coloring the vertices of~$G$ by $\{ 1, 2, 3 \}$.  Discard edges between vertices of the same color and make two copies of each vertex of color~$1$, the first connected to color-$2$ vertices and the second connected to color-$3$ vertices.  Connect $s$ to all the first vertices of color~$1$ and connect~$t$ to all the second vertices of color~$1$, and ask if~$s$ is connected to~$t$.  Unfortunately, this will give false positives.  If $G$ is a path of length four, with vertices colored $1, 2, 3, 1$, then~$s$ will be connected to~$t$ even though~$G$ does not contain a triangle.  To fix it, we can drop a breadcrumb at the first color-$1$ vertex and require that it be retrieved after the color-$3$ vertex; then the algorithm will no longer accept this path.  The algorithm still does not work, though, because it will accept a cycle of length five colored $1, 2, 3, 2, 3$.  Since the $st$-connectivity algorithm works in undirected graphs, it cannot see that the path backtracked from color~$3$ to color~$2$; graph minors can fool the algorithm.  What we can say is that in this particular case, the technique gives an $O(n)$-query and $\tilde O(n)$-time quantum algorithm that detects whether~$G$ contains a triangle or is acyclic, i.e., does not contain a triangle as a minor, under the promise that one of the two cases holds.

\paragraph{Organization.}

After some necessary background, in \refsec{st}, we present the algorithm for $st$-connectivity, and analyze its query complexity.  In \refsec{subgraph}, we define the subgraph/not-a-minor promise problem, and solve it for the cases when the subgraph is a subdivided star or a triangle.  In \refsec{K5}, though, we show that the technique does not work for arbitrary subgraphs.  In \refsec{efficient}, we present a framework for span program evaluation, and prove that the above algorithms can be implemented time efficiently.  Finally, in \refsec{startwo}, we generalize the reduction given above for path detection and give an $O(n)$-query quantum algorithm for detecting as a subgraph a star with two subdivided~legs.

\section{Preliminaries} \label{sec:prelim}

Let $[n]$ denote the set $\{1,\dots,n\}$.  Let $\cC(A)$ denote the range or column space of a matrix~$A$.

\subsection{Graph theory}

Let $K_n$ be the complete graph on~$n$ vertices, and let $K_{m,n}$ be the complete bipartite graph on $m$ and~$n$ vertices.  A star is a complete bipartite graph $K_{1,m}$, and a claw is the star~$K_{1,3}$.  All graphs we consider are simple.  

A graph $T$ is said to be a {\em subgraph} of a graph $G$, if $T$ can be obtained from $G$ by deleting edges and isolated vertices. $T$ is said to be a {\em minor} of $G$, if it can be obtained from $G$ by deleting and contracting edges, and deleting isolated vertices. Contracting an edge $(x, y)$ involves replacing $x$ and $y$ by a new vertex that is adjacent exactly to the union of the neighbors of $x$ and $y$.

There is an alternative way of describing the minor relation.  Let $H$ be a graph, and $\{V_x\}$, where~$x$ runs through all the vertices of $T$, be a collection of connected and pairwise disjoint subsets of the vertices of~$H$.  We write $H = M T$ if the following holds: there is an edge $(u, v)$ in $T$ if and only if there is an edge between a vertex of $V_x$ and a vertex of $V_y$ in~$H$.  If this holds, the sets $V_x$ are called the {\em branch sets} of $MT$.  A graph~$T$ is contained in $G$ as a minor if and only if some $MT$ is contained in $G$ as a subgraph.

\subsection{Quantum computation and span programs} \label{sec:span}

We are interested in both query and time complexity of quantum algorithms.  Query complexity measures only the number of queries to the input oracle, whereas time complexity measures the total number of gates.  For a survey of query complexity, refer to~\cite{BuhrmanDeWolf02querysurvey}.

We develop quantum algorithms based on span programs over the real numbers.  

\begin{definition}[Span program~\cite{KarchmerWigderson93span}] \label{t:spanprogramdef}
A span program $\cP = (n, d, \ket \target, V_{\text{free}}, \{V_{i,b}\})$ consists of a ``target" vector $\ket \target \in \R^d$ and finite sets $V_{\text{free}}$ and $V_{1,0}, V_{1,1}, \ldots, V_{n,0}, V_{n,1}$ of ``input" vectors from $\R^d$.  

To $\cP$ corresponds a boolean function $f_\cP : \{0,1\}^n \rightarrow \{0,1\}$, defined by $f_\cP(x) = 1$ if and only if $\ket \target$ lies in the span of the vectors in $V_{\text{free}} \cup \bigcup_{i=1}^n V_{i, x_i}$.  
\end{definition}

We say that the input vectors in $V_{i,b}$ are {\em labeled} by the value $b$ of the $i$th input variable~$x_i$.  
For an input $x = (x_i)\in\{0,1\}^n$, define the {\em available input vectors} as the vectors in $V_{\text{free}} \cup \bigcup_{i \in [n]} V_{i, x_i}$.  The other input vectors are called the {\em false input vectors}.  Let $V$, $V(x)$ and $\Vf$ be matrices having as columns the input vectors, the available input vectors and the free input vectors in $V_{\text{free}}$, respectively.  The span program evaluates to $1$ on input $x$ if and only if $\ket \target \in \cC(V(x))$.  

A useful notion of span program complexity is the {\em witness size}.  
\begin{itemize}\itemsep0pt
\item
If $\cP$ evaluates to $1$ on input $x$, a {\em witness} for this input is a pair of vectors $\ket w$ and $\ket{\wf}$ such that $\Vf \ket{\wf}+ V(x) \ket w = \ket \target$.  Its witness size is defined as $\norm{\ket w}{}^2$. 
\item
If $f_\cP(x) = 0$, then a witness for~$x$ is any vector $\ket{w'} \in \R^d$ such that $\braket{w'}{\target} = 1$ and $\ket{w'} \perp \cC(V(x))$.  Since $\ket \target \notin \spn(V(x))$, such a vector exists.  The witness size of $\ket{w'}$ is defined as $\|V^\adjoint \ket{w'}\|^2$.  This equals the sum of the squares of the inner products of $\ket{w'}$ with all false input vectors.  
\end{itemize}

The witness size of span program $\cP$ on input $x$, $\wsize(\cP, x)$, is defined as the minimal size among all witnesses for $x$.  For $\cD \subseteq \{0,1\}^n$, let 
\begin{equation}
\wsize_b(\cP,\cD) = \max_{x\in \cD: f_\cP(x) = b} \wsize(\cP, x) \enspace .  
\end{equation}
Then the witness size of $\cP$ on domain~$\cD$ is defined as
\begin{equation} \label{eqn:wsize}
\wsize(\cP,\cD) = \sqrt{\wsize_0(\cP,\cD)\wsize_1(\cP,\cD)} \enspace .  
\end{equation}
This is equivalent to the standard definition; see Eq.~(2.8) in~\cite{Reichardt09spanprogram_arxivandfocs}.

Span programs can be converted into quantum query algorithms: 

\begin{thm}[\cite{Reichardt09spanprogram_arxivandfocs, Reichardt10advtight}] \label{thm:span}
For any boolean function $f\colon \cD \to \{0,1\}$, with $\cD \subseteq \{0,1\}^n$, if $\cP$ is a span program computing $f$ on domain~$\cD$, then there exists a quantum algorithm that evaluates $f$ with two-sided bounded error using $O(\wsize(\cP, \cD))$ queries.  
\end{thm}

A proof is given in \refsec{spanAlgorithm}.  
Conversely, if $f$ can be evaluated by a bounded-error quantum algorithm that makes~$Q$ queries, then there is a span program for~$f$ with $O(Q)$ witness size~\cite{Reichardt09spanprogram_arxivandfocs}.  Thus, searching for good quantum query algorithms is equivalent to searching for span programs with small witness size.

\section{Span program and quantum query algorithm for \texorpdfstring{$st$}{st}-connectivity} \label{sec:st}

A key idea in our arguments will be a simple span program for deciding $st$-connectivity.  We show: 

\begin{thm}
\label{thm:st}
Consider the $st$-connectivity problem on a graph $G$ given by its adjacency matrix.  Assume there is a promise that if $s$ and $t$ are connected by a path, then they are connected by a path of length at most~$d$.  Then the problem can be decided in $O(n\sqrt{d})$ quantum queries.  
\end{thm}

In \refsec{efficient}, we will prove that this algorithm can be implemented in $\tilde O(n\sqrt{d})$ time and $O(\log n)$ space.

\pfstart
Define a span program $\cP$ using the vector space~$\R^n$, with the vertex set of $G$ as an orthonormal basis.  The target vector is $\ket t - \ket s$.  For each pair of distinct vertices $\{u, v\}$, order the vertices arbitrarily and add the input vector $\ket u - \ket v$ labeled by the presence of the edge~$(u, v)$, i.e., $\ket u - \ket v$ is available when the $(u,v)$ entry of the adjacency matrix is~$1$.  The edge orientations are not important since $\ket v - \ket u = -(\ket u - \ket v)$.  

When $s$ is connected to~$t$ in~$G$, let $t = u_0, u_1, \dots, u_m = s$ be a path between them of length $m \le d$.  All vectors $\ket{u_i} - \ket{u_{i+1}}$ are available, and their sum is $\ket t - \ket s$.  Thus the span program evaluates to~$1$.  The witness size is at most $m \leq d$.  

Next assume that $t$ and $s$ are in different connected components of $G$.  Define $\ket{w'}$ by $\langle w',u\rangle = 1$ if~$u$ is in the connected component of $t$, and 0 otherwise.  Then $\langle w', t-s \rangle = 1$ and $\ket{w'}$ is orthogonal to all available input vectors.  Thus the span program evaluates to~$0$ with $\ket{w'}$ a witness.  Since there are $O(n^2)$ false input vectors, and the inner product of each of them with $\ket{w'}$ is at most 1, the negative witness size is $O(n^2)$.  

$\cP$'s witness size is thus $O(n\sqrt{d})$.  By \refthm{span}, the problem's quantum query complexity is~$O(n \sqrt{d})$.
\pfend

It is easy to see that the problem's query complexity is $1$ if $d = 1$ and is $O(\sqrt{n})$ if $d = 2$.  If~$d \geq 3$, and~$d = O(1)$, then the algorithm of \refthm{st} is optimal, which can be seen by a reduction from the unordered search problem.  The algorithm is also optimal if $d = \Theta(n)$, again by the lower bound from~\cite{DurrHeiligmanHoyerMhalla04graph}.  

Observe that when $s$ and~$t$ are connected, the span program's witnesses correspond exactly to balanced unit flows from~$s$ to~$t$ in~$G$.  The witness size of a flow is the sum over all edges of the square of the flow across that edge.  If there are multiple simple paths from~$s$ to~$t$, then it is therefore beneficial to spread the flow across the paths in order to minimize the witness size.  The optimal positive witness size is the same as the resistance distance between~$s$ and~$t$, $R_{st} \leq d$, i.e., the effective resistance, or equivalently twice the energy dissipation of a unit electrical flow, when each edge in the graph is replaced by a unit resistor~\cite{DoyleSnell84electricnetworks}.  Spreading the flow to minimize its energy is the main technique used in the analysis of quantum query algorithms based on learning graphs~\cite{Belovs11triangle, Zhu11learninggraph, LeeMagniezSantha11learninggraph, BelovsLee11learninggraph}, for which this span program for $st$-connectivity can be seen to be the key subroutine.  Since the negative witness size is $O(n^2)$, the overall witness size is $O(n \sqrt{\max R_{st}})$, where the maximum is over allowed input graphs.  Notice that the hitting time from~$s$ to~$t$ for a classical random walk is at most $2 m R_{st}$, where~$m$ is the number of edges in~$G$~\cite{ChandraRaghavanRuzzoSmolenskyTiwari89resistance}.  A quantum walk that is given~$G$ achieves a square-root speedup in the hitting time~\cite{MagniezNayakRichterSantha08hittingtimes, Szegedy04walkfocs}; our algorithm is only slower by a factor of $O(n / \sqrt m)$, even though it is charged for accessing the input graph.

\section{Subgraph/not-a-minor promise problem} \label{sec:subgraph}

A natural strategy for deciding a minor-closed forbidden subgraph property is to take the list of forbidden subgraphs and test the input graph $G$ for each subgraph one by one.  Let $T$ be a forbidden subgraph from the list.  To simplify the problem of detecting~$T$, we can add the promise that $G$ either contains $T$ as a subgraph or does not contain~$T$ \emph{as a minor}.  Call this problem the subgraph/not-a-minor promise problem for~$T$.

In this section, we develop an approach to the subgraph/not-a-minor problem using span programs.  We first show that the approach achieves the optimal $O(n)$ query complexity in the case that~$T$ is a subdivided star.  Then in \refsec{triangle} we extend the approach to give an optimal $O(n)$-query algorithm for the case that~$T$ is a triangle.  In \refsec{K5}, however, we show that the approach fails for the case $T = K_5$.  

Before beginning, we state a lower bound that proves the optimality of these algorithms: 

\begin{prp} \label{prp:optimal}
If the graph $T$ has at least one edge, then the quantum query complexity of the subgraph/not-a-minor problem for $T$ is $\Omega(n)$, and the randomized query complexity is $\Omega(n^2)$.
\end{prp}

\pfstart
This is a standard argument by a reduction from the unordered search problem; see, e.g.,~\cite{BurhmanDurrHeiligmanHoyerMagniezSanthaDeWolf00elementdistinctness}.  Let $H$ be the smallest connected component of $T$ of size at least 2. Let $H'$ be $H$ with a vertex removed.  Let $G$ be constructed as $T \setminus H$ together with $n$ disjoint copies of $H'$ and $n$ isolated vertices.  The graph~$G$ has $O(n)$ vertices and does not contain a $T$-minor.  

Let $x_{i,j}$, for $i, j \in [n]$, be boolean variables.  Define $G(x)$ as $G$ with the $j$th isolated vertex connected to all vertices of the $i$th copy of $H'$, for all $i, j$ such that $x_{i,j} = 1$.  The graph $G(x)$ contains~$T$ as a subgraph if and only if at least one $x_{i,j}$ is~$1$.  This gives the reduction.  Unordered search on $n^2$ inputs requires $\Omega(n)$ quantum queries~\cite{BoyerBrassardHoyerTapp96search} and, clearly, $\Omega(n^2)$ randomized queries.  
\pfend

\subsection{Subdivision of a star} \label{sec:star}

In this section, we give an optimal quantum query algorithm for the subgraph/not-a-minor promise problem for a graph~$T$ that is a subdivided star.  As a special case, this implies an optimal quantum query algorithm for deciding minor-closed forbidden subgraph properties that are determined by a single forbidden subgraph.  

\begin{thm} \label{thm:star}
Let $T$ be a subdivision of a star.  Then there exists a quantum algorithm that, given query access to the adjacency matrix of a simple graph $G$ with $n$ vertices, makes $O(n)$ queries, and, with probability at least~$2/3$, accepts if $G$ contains~$T$ as a subgraph and rejects if $G$ does not contain $T$ as a minor.  
\end{thm}

It can be checked that if $T$ is a path or a subdivision of a claw then a graph $G$ contains $T$ as a minor if and only if it contains it as a subgraph.  Moreover, disjoint collections of paths and subdivided claws are the only graphs~$T$ with this property.  This implies the following corollary: 

\begin{cor} 
\label{cor:claw}
Assume $T$ is a path or a subdivision of a claw.  Then there exists a quantum algorithm that, given query access to the adjacency matrix of a simple graph $G$ with $n$ vertices, detects whether it contains~$T$ as a subgraph in $O(n)$ queries, except with error probability at most~$1/3$.  
\end{cor}

In \refsec{efficient}, we prove that the algorithms from \refthm{star} and \refcor{claw} can be implemented efficiently, in~$\tilde O(n)$ time and $O(\log n)$ space.  

\pfstart[Proof of \refthm{star}]
The proof uses the color-coding technique from~\cite{AlonYusterZwick95colorcoding}.  
Let $T$ be a star with~$d$ legs, of lengths $\ell_1, \ldots, \ell_d > 0$.  Denote the root vertex by~$r$ and the vertex at depth~$i$ along the $j$th leg by~$v_{j,i}$.  The vertex set of~$T$ is $V_T = \{ r, v_{1,1}, \ldots, v_{1, \ell_1}, \ldots, v_{d, 1}, \ldots, v_{d, \ell_d} \}$.  Color every vertex~$u$ of $G$ with an element $c(u) \in V_T$ chosen independently and uniformly at random.  For $v \in V_T$, let $c^{-1}(v)$ be its preimage set of vertices of~$G$.  We design a span program that 
\begin{itemize}
\item Accepts if there is a correctly colored $T$-subgraph in $G$, i.e., an injection $\iota$ from $V_T$ to the vertices of~$G$ such that $c \circ \iota$ is the identity, and $(t, t')$ being an edge of~$T$ implies that $(\iota(t), \iota(t'))$ is an edge of~$G$; 
\item Rejects if $G$ does not contain $T$ as a minor, no matter the coloring~$c$.
\end{itemize}
If $G$ contains a $T$-subgraph, then the probability it is colored correctly is at least $\abs{V_T}^{-\abs{V_T}} = \Omega(1)$.  Evaluating the span program for a constant number of independent colorings therefore suffices to detect~$T$ with probability at least~$2/3$.

\paragraph{Span program.}

The span program we define works on the vector space with orthonormal basis 
\begin{equation} \label{eqn:starspanprogrambasis}
\{ \ket s, \ket t \} \cup \Big\{ \ket{u, b} : (u, b) \in \big( c^{-1}(r) \times \{0, \ldots, d\} \big) \cup \bigcup_{v \in V_T \smallsetminus \{ r \}} c^{-1}(v) \times \{0,1\} \Big\} \enspace .
\end{equation}
The target vector is $\ket t - \ket s$.  For $u \in c^{-1}(r)$, there are free input vectors $\ket{u, 0} - \ket s$ and $\ket t - \ket{u, d}$.  For $j \in [d]$ and $u \in c^{-1}(v_{j, \ell_j})$, there are free input vectors $\ket{u, 1} - \ket{u, 0}$.  
For $j \in [d]$, there are the following input vectors: 
\begin{itemize}
\item 
For $i \in [\ell_j - 1]$, $u \in c^{-1}(v_{j,i})$ and $u' \in c^{-1}(v_{j,i+1})$, the input vectors $\ket{u', 0} - \ket{u, 0}$ and $\ket{u, 1} - \ket{u', 1}$ are available when there is an edge $(u, u')$ in~$G$.  
\item 
For $u \in c^{-1}(r)$ and $u' \in c^{-1}(v_{j,1})$, the input vector $(\ket{u', 0} - \ket{u, j-1}) + (\ket{u, j} - \ket{u', 1})$ is available when there is an edge $(u, u')$ in~$G$.  
\end{itemize}

For visualizing and arguing about this span program, it is convenient to define a graph~$H$ whose vertices are the basis vectors in Eq.~\refeqn{starspanprogrambasis}.  Edges of~$H$ correspond to the available span program input vectors; for an input vector with two terms, $\ket \alpha - \ket \beta$, add an edge $(\ket \alpha, \ket \beta)$, and for the four-term input vectors $(\ket{u', 0} - \ket{u, j-1}) + (\ket{u, j} - \ket{u', 1})$ add two ``paired" edges, $(\ket{u', 0}, \ket{u, j-1})$ and $(\ket{u, j}, \ket{u', 1})$.

\paragraph{Positive case.}

Assume that there is a correctly colored $T$-subgraph in $G$, given by a map $\iota$ from $V_T$ to the vertices of~$G$.  Then the target $\ket t - \ket s$ is achieved as the sum of the input vectors spanned by $\ket{s}$, $\ket{t}$ and the basis vectors of the form $|u,\cdot\rangle$ with $u\in \iota(V_T)$.  All these vectors are available.  This sum has a term $\ket \beta - \ket \alpha$ for each pair of consecutive vertices $\ket \alpha, \ket \beta$ in the following path from $\ket{s}$ to $\ket{t}$ in $H$:
\begin{multline*}
\ket{s}, \ket{\iota(r), 0}, \ket{\iota(v_{1,1}), 0}, \ket{\iota(v_{1,2}), 0}, \ldots, \ket{\iota(v_{1,\ell_1}), 0}, \ket{\iota(v_{1,\ell_1}), 1}, \ket{\iota(v_{1,\ell_1-1}), 1}, \ldots, \ket{\iota(v_{1,1}), 1}, \ket{\iota(r), 1}, \\ \ket{\iota(v_{2,1}), 0}, \ldots, \ket{\iota(v_{2,1}), 1}, \ket{\iota(r), 2}, \ket{\iota(v_{3,1}), 0}, \ldots \ldots, \ket{\iota(v_{d,1}), 1}, \ket{\iota(r), d}, \ket{t} \enspace .
\end{multline*}
 Pulled back to~$T$, the path goes from~$r$ out and back along each leg, in order.  The positive witness size is $O(1)$, since there are $O(1)$ input vectors along the path.  

\smallskip

This argument shows much of the intuition for the span program.  $T$ is detected as a path from~$\ket s$ to~$\ket t$, starting at a vertex in~$G$ with color~$r$ and traversing each leg of~$T$ in both directions, out and back.  It is not enough just to traverse~$T$ in this manner, though, because the path might each time use different vertices of color~$r$.  The purpose of the four-term input vectors $(\ket{u', 0} - \ket{u, j-1}) + (\ket{u, j} - \ket{u', 1})$ is to enforce that if the path goes out along an edge $(\ket{u, j-1}, \ket{u', 0})$, then it must return using the paired edge $(\ket{u', 1}, \ket{u, j})$.

\paragraph{Negative case.}

Assume that $G$ does not contain $T$ as a minor.  It may still be that $\ket s$ is connected to~$\ket t$ in~$H$.  We construct an ancillary graph $H'$ from $H$ by removing some vertices and adding some extra edges, so that $\ket s$ is disconnected from~$\ket t$ in~$H'$.  \reffig{examplesubdividedstar} shows an example.  

\begin{figure}
\vspace{-.4cm}
\centering
\begin{tabular}{c c}
\subfigure[\label{fig:exampleT}$T$]{$\qquad\qquad$\raisebox{.65cm}{\includegraphics[scale=1]{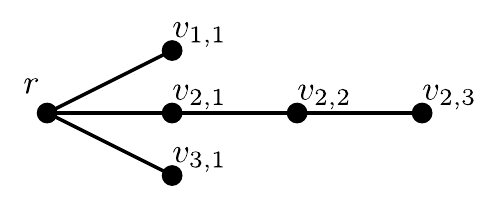}}} &
\subfigure[\label{fig:exampleG}$G$]{$\qquad\qquad$\includegraphics[scale=1]{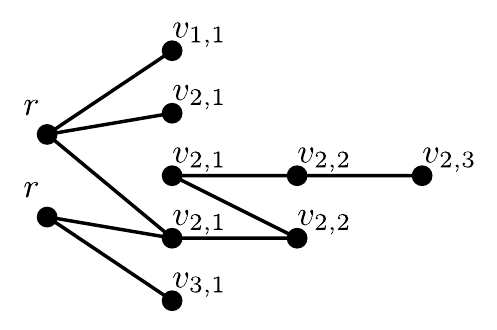}} \\
\subfigure[\label{fig:exampleH}$H$]{\includegraphics[scale=1]{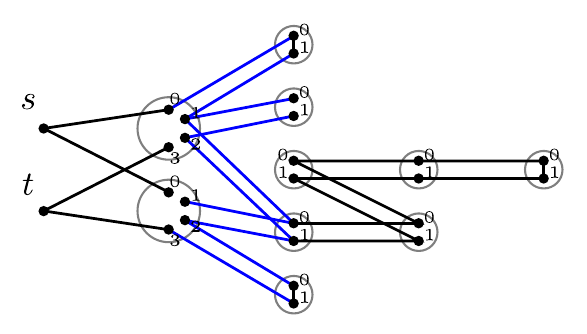}} &
\subfigure[\label{fig:exampleHprime}$H'$]{\includegraphics[scale=1]{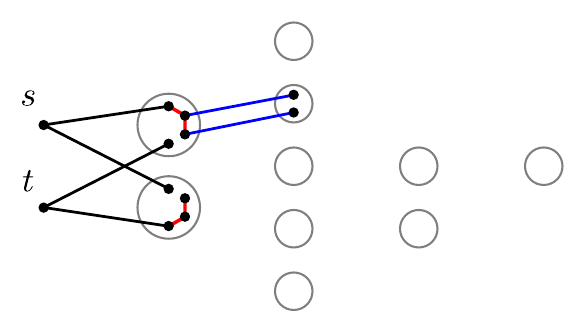}}
\end{tabular}
\caption{An example to illustrate the constructions of the graphs~$H$ and~$H'$ in the negative case of the proof of \refthm{star}.  (a)~A subdivided star $T$ with $(\ell_1, \ell_2, \ell_3) = (1, 3, 1)$.  (b)~A graph $G$ with vertices labeled by their colors, i.e., vertices of~$T$.  Although $G$ contains $T$ as a subgraph, the coloring is incorrect and the span program will reject.  (c)~The graph $H$, edges and paired edges of which correspond to available span program input vectors.  Vertices of $G$ have been split into two or four parts, and vertices~$s$ and~$t$ have been added.  Paired edges are colored blue.  Note that $s$ is connected to~$t$ in~$H$.  (d)~The graph~$H'$.  New edges are colored red.  Note that $s$ is disconnected from~$t$.  
} \label{fig:examplesubdividedstar}
\end{figure}

The graph $H'$ is defined starting with~$H$.  
Let $V_j = \{ v_{j,1}, \ldots, v_{j, \ell_j} \}$, $H_j = \{\ket{u, b} : c(u) \in V_j,\, b \in \{0,1\}\}$ and $R_j = \{ \ket{u,j} : c(u) = r \}$.  For $j \in [d]$ and $u \in c^{-1}(r)$, 
\begin{itemize}
\item Add an edge $(\ket{u, j-1}, \ket{u, j})$ to $H'$ if $\ket{u, j-1}$ is connected to $R_j$ in~$H$ via a path for which all internal vertices, i.e., vertices besides the two endpoints, are in $H_j$; and 
\item Remove all vertices in $H_j$ that are connected both to $R_{j-1}$ and $R_j$ in~$H$ via paths with all internal vertices in~$H_j$.  
\end{itemize}

Note that in the second case, for each $u' \in c^{-1}(V_j)$, either both $\ket{u', 0}$ and $\ket{u', 1}$ are removed, or neither~is.  Indeed, if there is a path from $\ket{u', 0}$ to $R_j$, then it necessarily must pass along an edge $(\ket{u'', 0}, \ket{u'', 1})$ with $c(u'') = v_{j, \ell_j}$.  Then backtracking along the path before this edge, except with the second coordinate switched $0 \leftrightarrow 1$, gives a path from $\ket{u', 0}$ to $\ket{u', 1}$.  Similarly $\ket{u', 0}$ is connected to $\ket{u', 1}$ if there is a path from $\ket{u', 1}$ to $R_{j-1}$.  

Define the negative witness $\ket{w'}$ by $\braket{v}{w'} = 1$ if $\ket s$ is connected to $\ket v$ in~$H'$, and $\braket{v}{w'} = 0$ otherwise.  Then $\ket{w'}$ is orthogonal to all available input vectors.  In particular, it is orthogonal to any available four-term input vector $(\ket{u', 0} - \ket{u, j-1}) + (\ket{u, j} - \ket{u', 1})$, corresponding to two paired edges in~$H$, because either the same edges are present in~$H'$, or $\ket{u', 0}$ and $\ket{u', 1}$ are removed and a new edge $(\ket{u, j-1}, \ket{u, j})$ is added.  

To verify that $\ket{w'}$ is a witness for the span program evaluating to~$0$, with $(\bra s - \bra t) \ket{w'} = 1$, it remains to prove that $\ket s$ is disconnected from $\ket t$ in~$H'$.  Assume that $\ket s$ is connected to~$\ket t$ in~$H'$, via a simple path~$p$.  Based on the path~$p$, we will construct a minor of~$T$ in~$G$, giving a contradiction.  

The path~$p$ begins at~$\ket s$ and next must move to some vertex $\ket{u_0, 0}$, where $c(u_0) = r$.  The path ends by going from a vertex $\ket{u_d, d}$, where $c(u_d) = r$, to $\ket t$.  By the structure of the graph~$H'$, $p$ must also pass in order through some vertices $\ket{u_1, 1}, \ket{u_2, 2}, \ldots, \ket{u_{d-1}, d-1}$, where $c(u_j) = r$.  

Consider the segment of the path from $\ket{u_{j-1}, j-1}$ to $\ket{u_j, j}$.  Due to the construction, this segment must cross a new edge added to~$H'$, $(\ket{u_j', j-1}, \ket{u_j', j})$ for some $u_j'$ with $c(u_j') = r$.  Thus the path~$p$ has the form 
\begin{equation*}
\ket s, \ldots, \ket{u_1', 0}, \ket{u_1', 1}, \ldots, \ket{u_2', 0}, \ket{u_2', 1}, \ldots \ldots, \ket{u_d', 0}, \ket{u_d', 1}, \ldots, \ket t \enspace .  
\end{equation*}
Based on this path, we can construct a minor for~$T$.  The branch set of the root~$r$ consists of all the vertices in~$G$ that correspond to vertices along~$p$ (by discarding the second coordinate).   
Furthermore, for each edge $(\ket{u_j', j-1}, \ket{u_j', j})$, there is a path in~$H$ from $\ket{u_j', j-1}$ to $R_j$, in which every internal vertex is in $H_j$.  The first $\ell_j$ vertices along the path give a minor for the $j$th leg of~$T$.  It is vertex-disjoint from the minors for the other legs because the colors are different.  It is also vertex-disjoint from the branch set of~$r$ because no vertices along the path are present in~$H'$.  Therefore, we obtain a minor for~$T$, a contradiction.  

Since each coefficient of $\ket{w'}$ is zero or one, the overlap of $\ket{w'}$ with any input vector is at most two in magnitude.  Since there are $O(n^2)$ input vectors, the witness size is $O(n^2)$.  

\smallskip

By Eq.~\refeqn{wsize}, the span program's overall witness size is the geometric mean of the worst witness sizes in the positive and negative cases, or $O(n)$.
\pfend

\begin{figure}
\vspace{-.4cm}
\centering
\begin{tabular}{c@{$\;$}c@{$\quad\!\!$}c}
\subfigure[$T$]{\raisebox{.95cm}{$\!\!\!\!\!\!\!\!$\includegraphics[scale=.9]{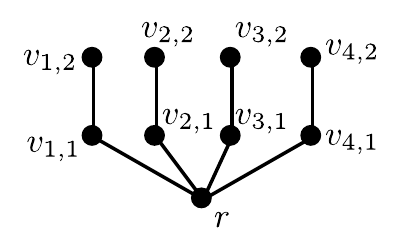}}}
&
\subfigure[$G$]{\raisebox{.95cm}{\includegraphics[scale=.9]{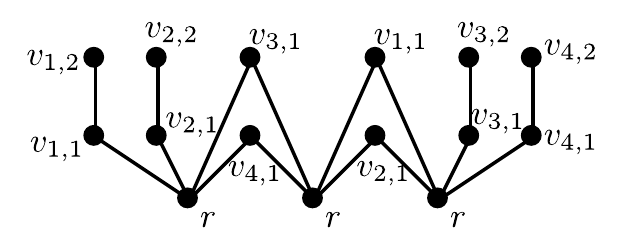}}}
&
\subfigure[$H$]{\includegraphics[scale=1.3]{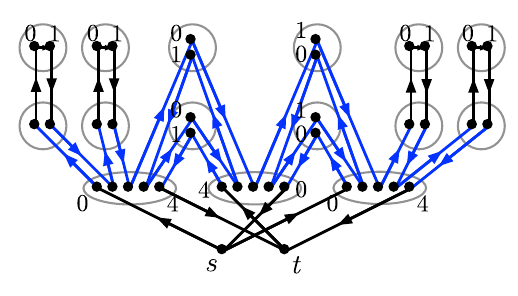}}
\end{tabular}
\caption{Let $T$ be the subdivided star with four legs of lengths $\ell_1 = \cdots = \ell_4 = 2$.  Then the span program from the proof of \refthm{star} accepts the colored graph~$G$ in~(b), even though $G$ contains~$T$ only as a minor and not as a subgraph.  The corresponding graph~$H$ is shown in~(c), together with a flow that indicates the combination of input vectors adding to $\ket t - \ket s$.  Notice that the flow is balanced at all vertices except~$s$ and~$t$, and also that the flows along paired edges are of equal strengths in opposite directions.  
} \label{fig:exampleacceptsminor}
\end{figure}

The promise that $G$ does not contain $T$ as a minor is necessary for the correctness of the algorithm; see~\reffig{exampleacceptsminor}.  

\begin{thm} \label{thm:starforest}
Let $T$ be a collection of vertex-disjoint subdivided stars.  Then there exists a quantum algorithm that, given query access to the adjacency matrix of a simple graph $G$ with $n$ vertices, makes $O(n)$ queries, and, with probability at least~$2/3$, accepts if $G$ contains~$T$ as a subgraph and rejects if~$G$ does not contain $T$ as a minor.  
\end{thm}

\pfstart
It is not enough to apply \refthm{star} once for each component of~$T$, because some components might be subgraphs of other components.  Instead, proceed as in the proof of \refthm{star}, but for each fixed coloring of~$G$ by the vertices of~$T$ run the span program once for every component on the graph $G$ restricted to vertices colored by that component.  This ensures that in the negative case, if the span programs for all components accept, then there are vertex-disjoint minors for every component, which together form a minor for~$T$.  
\pfend

Paired edges are more complicated to work with than ordinary edges.  For implementing the quantum algorithm time efficiently, in \refthm{efficient} below, we will therefore work with a slightly different span program in which the vertices $\ket{u, b}$, for $(u, b) \in c^{-1}(r) \times [d-1]$, are split in four and the vertices $\ket{u, b}$, for $(u, b) \in c^{-1}(r) \times \{0, d\}$ are split in two.  The modified span program computes the same function on allowed input graphs~$G$, with nearly the same witness size, but has the advantage that any vertex is incident to at most one paired edge.

\subsection{Triangle} \label{sec:triangle}

The technique used in the proof of \refthm{star} extends also to other problems.  As an example, we consider the case that~$T$ is a triangle.  Although the best known algorithm for detecting triangle subgraphs uses $O(n^{1.\overline{296}})$ queries~\cite{Belovs11triangle}, triangles can be detected in sparse graphs in $O(n^{1.1\overline{6}})$ queries~\cite[Theorem~4.4]{ChildsKothari10graphproperties}.  

\begin{thm} \label{thm:triangle}
There exists a $O(n)$-query quantum algorithm that, given query access to the adjacency matrix of a simple graph $G$ with $n$ vertices, accepts if $G$ contains a triangle and rejects if $G$ is a forest, i.e., does not contain a triangle as a minor, except with error probability at most~$1/3$.  
\end{thm}

\pfstart
The algorithm is similar to the one in \refthm{star}.  Let $c$ be a uniformly random map from the vertex set~$V_G$ of~$G$ to $\{ 0, 1, 2 \}$.  Define a span program on a vector space with orthonormal basis 
\begin{equation} \label{eqn:trianglespanprogrambasis}
\{ \ket s, \ket t \} \cup \{ \ket{u, c(u)} : u \in V_G \} \cup \{ \ket{u, 3} : u \in c^{-1}(0) \} \enspace .
\end{equation}
The target vector is $\ket t - \ket s$.  The free input vectors are $\ket t - \ket s + \ket{u, 0} - \ket{u, 3}$ for $u \in c^{-1}(0)$.  For $j \in \{0,1,2\}$ and $(u, u') \in c^{-1}(j) \times c^{-1}(j+1 \bmod 3)$, add an input vector $\ket{u',j+1} - \ket{u,j}$ that is available if the edge $(u, u')$ is present in~$G$.  

If $G$ contains a triangle, then the triangle is colored correctly with probability $2/9$.  Say the triangle is $\{u_0, u_1, u_2\}$, with $c(u_j) = j$.  Since the sum of the input vectors $\ket t - \ket s + \ket{u_0, 0} - \ket{u_0, 3}$, $\ket{u_1, 1} - \ket{u_0, 0}$, $\ket{u_2, 2} - \ket{u_1, 1}$ and $\ket{u_0, 3} - \ket{u_2, 2}$ equals $\ket t - \ket s$, the span program accepts.  The witness size is~$3$.  

The intuition for this construction is similar to \refthm{star}.  By using a four-term input vector $\ket t - \ket s + \ket{u, 0} - \ket{u, 3}$ for $u \in c^{-1}(0)$, instead of two separate input vectors $\ket{u, 0} - \ket s$ and $\ket t - \ket{u, 3}$, we prevent the span program from accepting paths $u_0, u_1, u_2, u_0'$ with $c(u_0') = 0$ but $u_0' \neq u_0$.  

Let us make this intuition precise.  Assume that $G$ is acyclic.  We argue that the span program rejects by constructing a negative witness~$\ket{w'}$.  Unlike in \refthm{star}, the coefficients of $\ket{w'}$ will not be only~$0$ or~$1$, and the worst-case witness size is $\Theta(n^4)$.  We will prove that the expected witness size is $O(n^2)$.  

Fix arbitrarily a root for every tree component of~$G$, and measure depths from these root vertices.  Let~$H$ be the same graph as~$G$, except with edges connecting vertices of the same color removed.  For every tree component in~$H$, set the root to be the (unique) vertex in that component with least depth in~$G$.  For a vertex~$u$, let $d(u)$ be its depth in~$H$.  
Observe that because $G$ is acyclic, going from $G$ to~$H$ every edge is removed independently with probability $1/3$.  Let $H'$ be the same as $H$ but with each vertex $u \in c^{-1}(0)$ split into two vertices $(u, 0)$ and $(u, 3)$, so that $(u, 0)$ is connected to $u$'s neighbors of color~$1$, and $(u, 3)$ is connected to $u$'s neighbors of color~$2$.  Also add an edge from $(u, 0)$ to $(u, 3)$.  $H'$ is acyclic.  

Using the graph $H'$, we can specify a negative witness $\ket{w'}$.  Let $\braket{s}{w'} = 1$ and $\braket{t}{w'} = 0$.  Since the vertices of $H'$ are in one-to-one correspondence with the other basis vectors of Eq.~\refeqn{trianglespanprogrambasis}, it remains to give coefficients for each vertex of~$H'$.  Note that for any $u \in c^{-1}(0)$, the condition that $\ket{w'}$ be orthogonal to the free input vector $\ket t - \ket s + \ket{u, 0} - \ket{u, 3}$ implies that $\ket{w'}$ must satisfy $\braket{u, 0}{w'} = \braket{u, 3}{w'} + 1$.  Up to an additive factor, this condition determines the coefficients of $\ket{w'}$ for each connected component of~$H'$.  Let $r$ be the root of the component. For a vertex $u$ in the component, define the level $\ell(u)$ as the number of $((u,3),(u,0))$ edges minus the number of $((u,0),(u,3))$ edges traversed along the simple path from $r$ to~$u$.  Let $\braket{u}{w'} = \ell(u)$.  Note that $\ell(u) \leq d(u)+1$ because no two new edges are adjacent.  

Unfortunately, the coefficients of~$\ket{w'}$ may grow as large as $\Omega(n)$, resulting in a negative witness size of order $n^4$. However, the probability of this event is negligible. Indeed, the negative witness size is bounded by
\begin{equation*}
\sum_{u, v \in H'} \langle u-v, w' \rangle^2 \le \sum_{u, v \in H'} 2(\braket{u}{w'}^2 + \braket{v}{w'}^2) \le 4n \sum_{v\in H'} (d(v)+1)^2 \enspace ,
\end{equation*}
because $H'$ has at most $2n$ vertices.  For a fixed $v$, the expectation of $(d(v)+1)^2$ is bounded by the series $\frac13 \sum_{i=0}^\infty (i+1)^2(2/3)^i = O(1)$.  By linearity of expectation, the expected size of the negative witness is~$O(n^2)$.  By a Markov inequality, for any $\eps > 0$ one may choose $C$ so that the probability the negative witness size exceeds $C n^2$ is less than~$\eps$.  This case adds at most~$\eps$ to the algorithm's error probability.  If the negative witness size is at most $C n^2$ then the total witness size is~$O(n)$.  
\pfend

\subsection{A counterexample for \texorpdfstring{$K_5$}{K\_5}} \label{sec:K5}

The algorithms in Sections~\ref{sec:star} and~\ref{sec:triangle} suggest a general approach for solving the subgraph/not-a-minor problem for a graph $T$: randomly color $G$ by the vertices of~$T$, and construct a span program for a traversal of~$H$, using the paired-edge trick to assure that the same vertex of $G$ is chosen for all appearances of a vertex of $T$ in the traversal.  Natural candidate graphs to consider next include general trees and cycles.  In this section, however, we show that the approach fails for some graphs~$T$.  

\begin{figure}
\vspace{-.4cm}
\[
\def\objectstyle{\scriptstyle}
\raisebox{-1.35cm}{\includegraphics{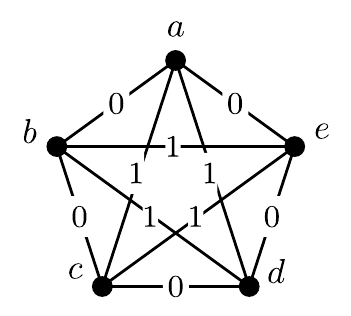}}
\quad \times \quad \Z/2\Z \quad = \qquad 
\raisebox{-1.75cm}{\includegraphics{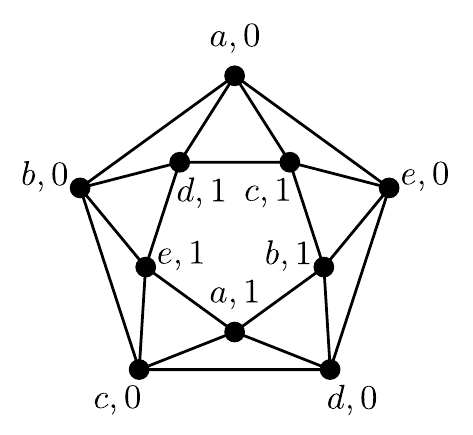}}
\]
\caption{A skew product of $K_5$ and $\Z/2\Z$ gives a planar graph that does not contain~$K_5$ as a minor.  This example is due to Jim Geelen.}
\label{fig:skew}
\end{figure}

Consider the following operation that is a special case of the skew product of a graph and a group~\cite{KumjianPask99skew}.  Let $T$ be a graph with each edge $e$ marked by $s_e \in \Z/2\Z$.  The skew product of $T$ and $\Z/2\Z$ is the graph $T_2$ with vertices $(v, i)$, where $v$ is a vertex of~$T$ and $i \in \Z/2\Z$.  $T_2$ has two edges for each edge $(u, v)$ of~$T$: $\big( (u,i), (v,i+s_{(u,v)}) \big)$ for $i \in \Z/2\Z$.  \reffig{skew} shows an example.  

The span program built along the lines of the algorithm from Theorems~\ref{thm:star} and~\ref{thm:triangle} accepts on this graph if it is colored correctly, i.e., if both vertices $(v,0)$ and $(v,1)$ of $T_2$ are colored by~$v$ in~$T$.  Indeed, the positive witness for $G = T_2$ can use all available input vectors with uniform coefficients~$1/2$.  

In general, however, and as shown in \reffig{skew}, $T_2$ does not contain $T$ as a minor.  It is easy to check that if $T$ is a tree or a triangle, $T_2$ does contain $T$ as a minor---and even as a subgraph, in the case of a tree.  

This shows that our algorithm does not work for all subgraph/not-a-minor promise problems.  Similarly, one can define a (total) minor-closed forbidden subgraph property for which our algorithm fails.  The property of having as a minor neither $K_5$ nor the eleven-vertex path~$P_{11}$ is a forbidden subgraph property.  The product graph in \reffig{skew} satisfies this property, but our algorithm will falsely detect a $K_5$ subgraph.  

Characterizing the quantum query complexities of minor-closed forbidden subgraph properties is an interesting problem.  Does any minor-closed forbidden subgraph property have $\omega(n)$ quantum query complexity?

\section{Time-efficient implementations} \label{sec:efficient}

A span program~$\cP$, on domain $\cD$, can be evaluated by a quantum algorithm that only makes $O(\wsize(\cP, \cD))$ queries to the input string (\refthm{span}).  The algorithm alternates a fixed, input-independent reflection with a simple input-dependent reflection.  This structure is inherited from Grover's search algorithm~\cite{Grover96search}.  In general, however, the algorithm will not be time efficient, because the input-independent reflection will be difficult to implement using local gates.  The time-efficient implementation of span programs can be subtle.  

In this section, we show how to use a quantum walk to implement efficiently the input-independent reflection for the algorithms of Theorems~\ref{thm:st}, \ref{thm:starforest} and~\ref{thm:triangle}.  Roughly, the quantum walk is on either the complete graph or a layered graph with complete bipartite graphs between adjacent layers.  Some modifications are needed, however, to deal with paired edges.  The desired reflection is about the stationary eigenspace of the quantum walk.  The graph's constant spectral gap allows for implementing this reflection to within inverse polynomial precision using only logarithmically many steps of the walk.  The graph's uniform structure allows for implementing each step efficiently.  We will show: 

\begin{thm} \label{thm:efficient}
The algorithm from \refthm{st} can be implemented in $\tilde O(n \sqrt d)$ quantum time, and the algorithms from Theorems~\ref{thm:starforest} and~\ref{thm:triangle} can be implemented in $\tilde O(n)$ quantum time.  In these implementations, the algorithms from Theorems~\ref{thm:st} and~\ref{thm:starforest} use $O(\log n)$ bits and qubits of space.  
\end{thm}

Intuitively, our span program-based algorithms are similar to running a quantum walk on the input graph~$G$.  However, $G$ is given only by an input oracle, and implementing a quantum walk on it directly would require too many input queries.  Instead, we run a quantum walk on a nearly complete graph that contains~$G$, and interpolate input queries in order, roughly, to simulate a walk on~$G$.  

In the proof of \refthm{efficient}, we need some basic facts about {\em $k$-wise independent hash functions}; see, e.g.~\cite{LubyWigderson95hash}.  This is a collection of functions $h_m: [n] \to [\ell]$ such that for any~$k$ distinct elements $a_1,\dots,a_k$, the probability over the choice of~$m$ that $(h_m(a_1),\dots,h_m(a_k))$ takes a particular value in $[\ell]^k$, is $\ell^{-k}$.  The simplest construction, that suffices for our purposes, is to assume $\ell \le n$ are powers of two, and define~$h_m$ as the $\log_2\ell$ lowest bits of the value of a random polynomial over $GF(n)$ of degree $k-1$.  Then $O(k \log n)$ bits suffice to specify~$h_m$, from which $h_m(a)$ can be calculated in $O(k \log^2 n)$ time.  

We will also need some further background in linear algebra.

\subsection{Linear algebra background} \label{sec:linearalgebra}

Results about the product of two reflections have many applications in quantum algorithms.  Let $A$ and~$B$ be matrices each with $n$ rows and orthonormal columns.  Let $\Pi_A = A A^\adjoint$ and $\Pi_B = B B^\adjoint$ be the projections onto $\cC(A)$ and $\cC(B)$, respectively.  Denote by $R_A = 2\Pi_A-I$ and $R_B = 2\Pi_B-I$ the reflections about the corresponding subspaces, and let $U = R_B R_A$ be their product.  Let $D(A, B) = A^\adjoint B$.

\begin{lem}[Spectral Lemma~\cite{Szegedy04walkfocs, Jordan75projections}] \label{lem:szegedy}
Under the above assumptions, all the singular values of $D(A,B)$ are at most~$1$.  Let $\cos\theta_1, \ldots, \cos\theta_\ell$ be all the singular values of $D(A,B)$ lying in the open interval~$(0,1)$, counted with multiplicity.  Then the following is a complete list of the eigenvalues of $U$: 
\begin{itemize}
\item The $+1$ eigenspace is $(\cC(A)\cap \cC(B))\oplus(\cC(A)^\perp \cap \cC(B)^\perp)$.
\item The $-1$ eigenspace is $(\cC(A)\cap \cC(B)^\perp)\oplus(\cC(A)^\perp \cap \cC(B))$.  Moreover, $\cC(A)^\perp \cap \cC(B) = B(\ker D(A,B))$.  
\item On the orthogonal complement of the above subspaces, $U$ has eigenvalues $e^{2 i \theta_j}$ and $e^{-2 i \theta_j}$ for $j\in[\ell]$.
\end{itemize}
\end{lem}

A consequence of the Spectral Lemma is: 

\begin{lem}[Effective Spectral Gap Lemma~\cite{LeeMittalReichardtSpalekSzegedy11stateconversion}] \label{lem:effective}
For $\Theta \geq 0$, let $P_\Theta$ be the orthogonal projection to the span of all eigenvectors of $U$ with eigenvalues $e^{i \theta}$ such that $\abs{\theta} \le \Theta$.  Then for $\ket u \in \cC(A)^\perp$, 
\begin{equation}
\norm{P_\Theta \Pi_B \ket u} \le \frac{\Theta}{2} \norm{\ket u} \enspace .
\end{equation}
\end{lem}

We will also use the following simple fact about the spectra of block matrices.  For $n \in \N$, let $I_n$ be the $n \times n$ identity matrix, and $J_n$ the $n \times n$ all-ones matrix.  

\begin{clm} \label{clm:blockmatrixspectrum}
Fix $\ell \times \ell$ symmetric matrices $A$ and~$B$.  For $n \in \N$, let $M_n = A \otimes I_n + \frac{1}{n} B \otimes J_n$.  Then the spectrum of~$M_n$, i.e., the set of eigenvalues sans multiplicities, is independent of~$n$.  
\end{clm}

\begin{proof}
Let $\{ \ket{u_i} : i \in [n] \}$ be an orthonormal eigensystem for $\tfrac{1}{n} J_n$, with corresponding eigenvalues $\lambda_i \in \{0, 1\}$.  For $i \in [n]$, let $M(i) = A + \lambda_i B$.  If $\ket v$ is an eigenvalue-$\lambda$ eigenvector of $M(i)$, then $\ket v \otimes \ket{u_i}$ is an eigenvalue-$\lambda$ eigenvector of $M_n$.  These derived eigenvectors span the whole $(\ell n)$-dimensional space, and hence the set of eigenvalues of $M_n$ does not depend on~$n$.  
\end{proof}

Essentially, the above argument works because $I_n$ and~$J_n / n$ commute and have spectra independent of~$n$.

\subsection{Algorithm for evaluating a span program} \label{sec:spanAlgorithm}

For evaluating our span programs, we essentially use Algorithm~$1$ from~\cite{Reichardt10advtight}.  However, this algorithm is described for canonical span programs only.  A canonical span program is a special case of a span program, literally corresponding to the dual of the adversary bound of a boolean function~\cite[Lemma 6.5]{Reichardt09spanprogram_arxivandfocs}.  Although it is known that any span program can be reduced to canonical form without cost to the witness size~\cite[Theorem 5.2]{Reichardt09spanprogram_arxivandfocs}, general span programs can be easier to work with both in constructing the span program and in developing a time-efficient implementation.  None of the span programs in this paper are canonical.  For completeness, we restate the algorithm and prove its correctness for general span programs.  The proof uses the Effective Spectral Gap Lemma from~\cite{LeeMittalReichardtSpalekSzegedy11stateconversion}.  

The free input vectors can be eliminated from any span program without affecting the witness size~\cite[Prop.~4.10]{Reichardt09spanprogram_arxivandfocs}.  They can be useful for implementing the algorithm time efficiently, however, but then must be charged for properly, as in the ``full witness size" complexity measure from~\cite{Reichardt09unbalancedformula}.  To do so, convert free input vectors to normal input vectors that are associated with an additional input variable~$x_0$ that is fixed to~$1$.  It will also be convenient to have some input vectors that are never available, also associated to~$x_0$.  Henceforth, we do not allow for free input vectors, and both always- and never-available input vectors are charged for in the witness size.   

Let $\cP$ be a span program with the target vector $\ket \target$ and $m - 1$ input vectors $\{ \ket{v_j} \}$ in $\R^d$.  Let $W_1$ and~$W_0$ be the positive and the negative witness sizes, respectively, and let $W = \sqrt{W_0 W_1}$ be the witness size of~$\cP$.  Also let $\ket{\tilde \target} = \ket \target / \alpha$, where $\alpha = C_1 \sqrt{W_1}$ for some constant~$C_1$ to be specified later.  

Let $V$ be the matrix containing the input vectors of $\cP$, and also $\ket{\tilde \target}$, as columns.  Our quantum algorithm works in the vector space $\cH = \R^m$, with standard basis elements $\ket{j}$ for $j = \{0, \ldots, m-1\}$.  Basis vectors~$\ket j$ for $j > 0$ correspond to the input vectors, and $\ket 0$ corresponds to $\ket{\tilde \target}$.  Let $\Lambda$ be the orthogonal projection onto the nullspace of $V$.  For any input $x$ of~$\cP$, let $\Pi_x = \sum \ketbra{j}{j}$ where the summation is over $j = 0$ and those indices $j > 0$ corresponding to the available input vectors on input~$x$.  

Let $U = R_\Lambda R_\Pi$, where $R_\Lambda = 2 \Lambda - I$ and $R_\Pi = 2 \Pi_x - I$ are the reflections about the images of $\Lambda$ and~$\Pi_x$.  Starting in $\ket 0$, the algorithm runs phase estimation~\cite{Kitaev95phaseestimation} on $U$ with precision $\Theta = \frac1{C_2 W}$, and accepts if and only if the measured phase is zero.  Here $C_2$ is another constant to be specified.  

\begin{thm}
Assume $C_1 W \ge 1$.  Then the above algorithm is correct and requires $O(W)$ controlled applications of $U$.  In each of these applications, $R_\Lambda$ requires no access to the input oracle, whereas $R_\Pi$ can be implemented in one oracle query.  
\end{thm}

\pfstart
The statements apart from correctness are trivial.  The number of applications of~$U$ is equal to the inverse of the precision, up to a constant factor~\cite{NagajWocjanZhang09qma}.  Note also that the extra space required for phase estimation, beyond the space needed to implement~$U$, is logarithmic in the inverse precision.  

Assume that $f(x) = 1$.  In this case, we have to show there is a unit-length, eigenvalue-one eigenvector~$\ket u$ of~$U$ having a large overlap with $\ket 0$.  Let $\ket w$ be an optimal witness for $x$, and let $\ket{\tilde u} = \alpha \ket 0 - \sum_j w_j \ket j$, where~$w_j$ is the witness coefficient for the $j$th input vector.  Then $R_\Pi \ket{\tilde u} = \ket{\tilde u}$, because $\ket w$ uses available input vectors only.  Also, $V \ket{\tilde u} = \alpha \ket{\tilde t} - \sum_j w_j \ket{v_j} = \ket \target - \ket \target = 0$, and hence, $R_\Lambda \ket{\tilde u} = \ket{\tilde u}$.  Thus, $\ket{\tilde u}$ is an eigenvalue-one eigenvector of $U$.  Note that $\| \sum_j w_j \ket j \|^2 \le W_1 = \alpha^2 / C_1^2$; hence, $\ket u = \ket{\tilde u} / \norm{\ket{\tilde u}}$ has a large overlap with $\ket 0$ that can be tuned by adjusting the value of $C_1$.  

Now assume $f(x) = 0$.  Let $P_\Theta$ be the projection on the span of the eigenvalues of $U$ with eigenvalues $e^{i \theta}$ such that $\abs{\theta} \le \Theta$.  We have to prove that $\|P_\Theta \ket 0\|$ is small.  The idea is to apply \reflem{effective}.  Let~$\ket{w'}$ be an optimal witness for $x$.  Let $\ket u = \alpha V^\adjoint \ket{w'}$.  Since $\ket u \in \cC(V^\adjoint)$, we have $\Lambda \ket u = 0$.  Also, $\ket{w'}$ is orthogonal to all available input vectors, and $\alpha \braket{0}{w'} = 1$; hence $\Pi_x \ket u = \ket 0$.  By \reflem{effective}, 
\[
\| P_\Theta \ket 0 \| 
= \| P_\Theta \Pi_x \ket u \| 
\le \frac{\Theta}{2} \|\ket u\| 
\le \frac{\sqrt{1 + \alpha^2 W_0}}{2 C_2 W} 
\le \frac{C_1 \sqrt{W_1 W_0}}{C_2 W} 
= \frac{C_1}{C_2} \enspace .
\]
The algorithm's acceptance probability can be improved by increasing the value of~$C_2$.  
\pfend

\subsection{Implementing \texorpdfstring{$R_\Lambda$}{R\_Lambda}}

The reflection $R_\Pi$ can be implemented efficiently in most cases, but implementing $R_\Lambda$ is more difficult.  Since many functions have larger time complexity than query complexity, this should be expected.  In this section, we describe a general way of implementing $R_\Lambda$, which is efficient for relatively uniform span programs like those in this paper.  

Essentially, we consider the $d \times m$ matrix $V$ as the biadjacency matrix for a bipartite graph on $d + m$ vertices, and run a Szegedy-type quantum walk as in~\cite{AmbainisChildsReichardtSpalekZhang07andor, ReichardtSpalek08spanprogram}.  Such a quantum walk requires ``factoring" $V$ into two sets of unit vectors, vectors $\ket{a_i} \in \R^m$ for each row $i \in [k]$, and vectors $\ket{b_j} \in \R^d$ for each column $j = 0, \ldots, m-1$, satisfying $\braket{a_i}{j} \braket{i}{b_j} = V'_{ij}$, where $V'$ differs from~$V$ only by a rescaling of the rows.  (In general, multiplying $V$ from the left by any non-degenerate matrix, and in particular rescaling its rows, does not affect the nullspace.)  Given such a factorization, let $A = \sum_{i \in [d]} (\ket i \otimes \ket{a_i}) \bra i$ and $B = \sum_{j=0}^{m-1} (\ket{b_j} \otimes \ket j) \bra j$, so $A^\dagger B = V'$.  Let $R_A$ and $R_B$ be the reflections about the column spaces of $A$ and~$B$, respectively.  
Embed $\cH$ into $\tcH = \R^k \otimes \R^m$ using the isometry $B$.  Then $R_\Pi$ can be implemented on $B(\cH) = \cC(B)$ by detecting~$j$ from the representation of $\ket i \otimes \ket j$ and multiplying the phase by~$-1$ if $\ket{v_j}$ is an unavailable input vector.  $R_\Lambda$ can be implemented on $B(\cH)$ as the reflection about the $-1$ eigenspace of~$R_B R_A$.  Indeed, by \reflem{szegedy}, this eigenspace equals $(\cC(A)^\perp \cap \cC(B)) \oplus (\cC(A) \cap \cC(B)^\perp)$, or $B (\ker V)$ plus a part that is orthogonal to $\cC(B)$ and therefore irrelevant.  

The reflection about the $-1$ eigenspace of $R_B R_A$ is implemented using phase estimation.  The efficiency depends on two factors: 
\begin{enumerate}
\item The implementation costs of $R_A$ and $R_B$.  They can be easier to implement than $R_\Lambda$ directly, because they decompose into local reflections.  The reflection $R_A$ about the columns of~$A$ equals a reflection about $\ket{a_i}$ controlled by the column~$i$, and similarly for~$R_B$.  
\item The spectral gap around the $-1$ eigenvalue of $R_B R_A$ necessary to implement the reflection about the $-1$ eigenspace.  By \reflem{szegedy}, this gap is determined by the spectral gap of $D(A, B) = A^\adjoint B = V'$ around singular value zero.
\end{enumerate}

So far the arguments have been general.  Let us now specialize to the span programs in Theorems~\ref{thm:st}, \ref{thm:star} and~\ref{thm:triangle}.  These span programs are sufficiently uniform that neither of the above two factors is a problem.  Both reflections can be implemented efficiently, in poly-logarithmic time, using quantum parallelism.  Similarly, we can show that $D(A, B)$ has an $\Omega(1)$ spectral gap around singular value zero.  Therefore, approximating to within an inverse polynomial the reflection about the $-1$ eigenspace of $R_B R_A$ takes only poly-logarithmic~time.  

\pfstart[Proof of \refthm{efficient}]
We give the proof for the algorithms from Theorems~\ref{thm:star} and~\ref{thm:triangle}.  The argument for $st$-connectivity, \refthm{st}, is similar and actually easier.  

Both algorithms look similar.  In each case, the span program is based on a graph~$H$, whose vertices form an orthonormal basis for the span program vector space.  The vertices of~$H$ can be divided into a sequence of layers that are monochromatic according to the coloring~$c$ induced from~$G$, such that edges only go between consecutive layers.  Precisely, place the vertices $s$ and~$t$ each on their own separate layer at the beginning and end, respectively, and set the layer of a vertex~$v$ to be the distance from~$s$ to $c(v)$ in the graph~$H$ for the case that $G = T$.  For example, in the span program for detecting a subdivided star with branches of lengths $\{ \ell_1, \ldots, \ell_d \}$, there are $\ell = 2 + 2 \sum_{j \in [d]} (\ell_j + 1)$ layers, because the $s$-$t$ path is meant to traverse each branch of the star out and back.  There are $\ell = 6$ layers of vertices for the triangle-detection span program.  

In order to facilitate finding factorizations~$\{ \ket{a_i} \}$ and~$\{ \ket{b_j} \}$ such that~$R_A$ and~$R_B$ are easily implementable, we make two modifications to the span programs.  

First, the span programs as presented depend on the random coloring of~$G$.  This dependence makes it difficult to specify a general factorization of~$V$.  To fix this, if~$G$ has $n$ vertices, add dummy vertices to every layer of the graph so that every layer has size~$n$.  Fill in the graph with never-available edges between adjacent layers, including between the layers of~$s$ and~$t$, so that every vertex has degree~$2 n$.  If the edges in two layers are paired, then also pair corresponding newly added edges; each edge pair corresponds to one never-available, four-term input vector.  This transformation is equivalent to making the coloring part of the input, in the following sense: if $v$ and~$v'$ are vertices in adjacent layers, corresponding to vertices~$u$ and~$u'$ of~$G$, then the $(v, v')$ edge input vector is available if $(u, u')$ is an edge in~$G$ \emph{and} if $u$ and~$u'$ are both colored appropriately.

Second, scale the input vectors corresponding to paired edges down by a factor of $\sqrt 2$.  
Connect~$s$ and~$t$ by \emph{two} edges, the first corresponding to the scaled target vector $\ket{\tilde \target} = \frac{1}{\alpha}(\ket t - \ket s)$, and the second a never-available input vector $\sqrt{1 - 1/\alpha^2} (\ket t - \ket s)$.  We may assume that $\alpha = C_1 \sqrt{W_1} \geq 1$.  

It is easy to verify that the span program after this transformation still computes the same function, and the positive and the negative witness sizes remain $O(1)$ and $O(n^2)$, respectively.  After the modifications, the graph~$H$ has a simple uniform structure that allows for facile factorization.  There is a complete bipartite graph between any two adjacent layers.  

We specify a vector $\ket{a_i}$ for each vertex~$i$ of the graph.  For $i \notin \{s, t\}$, let~$\ket{a_i}$ be the vector with uniform $1/\sqrt{2n}$ coefficients for all incident edges.  For $i \in \{s, t\}$, let $\ket{a_i}$ have coefficients $1/(\alpha\sqrt{2n})$ and $\sqrt{(1-1/\alpha^2)/(2n)}$ for the two edges between~$s$ and~$t$, and coefficients $1/\sqrt{2n}$ for the other $2n-1$ edges.  
For any edge or pair of paired edges---that is, for each of the input vectors and the target vector---we specify a vector~$\ket{b_j}$.  For an ordinary edge~$j$, let $\ket{b_j}$ be the vector with $\frac{1}{\sqrt 2}(1, -1)$ coefficients on the vertices connected by~$j$ and zeros elsewhere.  For a pair of paired edges corresponding to the input vector $\ket{v_j} = \frac{1}{\sqrt 2}(\ket{i} + \ket{i'} - \ket{i''} - \ket{i'''})$, let $\ket{b_j} = \frac{1}{\sqrt 2} \ket{v_j}$.  Then these $\ket{a_i}$ and $\ket{b_j}$ vectors give a factorization of $V' = \frac{1}{2 \sqrt n} V$, i.e., $\braket{a_i}{j} \braket{i}{v_i} = \frac{1}{2 \sqrt n} V_{i,j} = \frac{1}{2 \sqrt n} \braket{i}{v_j}$.  

Let us analyze the spectral gap around zero of $D(A, B) = A^\adjoint B = V'$.  The non-zero singular values of~$V'$ are the square roots of the non-zero eigenvalues of $\Delta =  V' V'^\dagger = \sum_{i, i'} \big( \frac{1}{4n} \sum_j \braket{i}{v_j} \braket{v_j}{i'} \big) \ketbra{i}{i'}$.  We need to compute~$\Delta$.  A vertex~$i$ can be represented by a tuple~$(k, \sigma) \in [\ell] \times [n]$, where $k$ specifies one of the~$\ell$ layers and~$\sigma$ specifies a vertex within the layer.  Let $\Delta(k, k')$ be the $n \times n$ submatrix of~$\Delta$ between vertices at layers~$k$ and~$k'$.  To calculate $\Delta(k, k')$, we consider separately the contributions from all of the different layers of input vectors.  

\begin{enumerate}
\item
Ordinary edges between adjacent layers~$k$ and~$k'$ contribute $\frac{1}{4} I_n$ to $\Delta(k,k)$ and $\Delta(k',k')$, and $-\frac{1}{4n} J_n$ to $\Delta(k,k')$ and $\Delta(k',k)$.  Indeed, for the contribution to $\Delta(k,k)$, observe that any vertex $(k, \sigma)$ has~$n$ incident ordinary edges to layer~$k'$, and each incident edge~$j$ contributes a term $\frac{1}{4n} \abs{\braket{i}{v_j}}^2 = \frac{1}{4n}$.  There is no ordinary edge involving vertices $(k,\sigma)$ and $(k,\sigma')$ with $\sigma \neq \sigma'$, but for any $\sigma, \sigma' \in [n]$, there is exactly one ordinary edge~$j$ from $(k, \sigma)$ to $(k', \sigma')$, and it contributes $-\frac{1}{4n}$ to $\Delta(k, k')_{\sigma, \sigma'}$.  

Even though~$s$ and~$t$ are connected by two edges, the same calculations hold for the edges between their layers.  

\item
Consider a set of paired edges, that go out from layer~$k_1$ to~$k_2$, and then return from layer~$k_3$ to~$k_4$.  Each input vector~$\ket{v_j}$ is of the form $\frac{1}{\sqrt 2} \big(- \ket{(k_1, \sigma)} + \ket{(k_2, \sigma')} - \ket{(k_3, \sigma')} + \ket{(k_4, \sigma)}\big)$.  The four layers $k_1, \ldots, k_4$ are distinct.  The contributions of these paired edges to the sixteen blocks $\Delta(k_\alpha, k_\beta)$ are given by the $4 \times 4$ block matrix 
\begin{equation*}
\begin{pmatrix}
\;\;\;\;\; \frac{1}{8} I_n 	& \!\!\! -\frac{1}{8n} J_n 		& \!\!\! \;\;\, \frac{1}{8n} J_n 	& \!\!\! \;\, -\frac{1}{8} I_n \\[1pt]
-\frac{1}{8n} J_n 		& \!\!\! \;\;\;\;\; \frac{1}{8} I_n 	& \!\!\! \;\, -\frac{1}{8} I_n 		& \!\!\! \;\;\, \frac{1}{8n} J_n \\[1pt]
\;\;\, \frac{1}{8n} J_n 	& \!\!\! \;\, -\frac{1}{8} I_n 		& \!\!\! \;\;\;\;\; \frac{1}{8} I_n 	& \!\!\! -\frac{1}{8n} J_n \\[1pt]
\;\, -\frac{1}{8} I_n 		& \!\!\! \;\;\, \frac{1}{8n} J_n 	& \!\!\! -\frac{1}{8n} J_n 		& \!\!\! \;\;\;\;\; \frac{1}{8} I_n
\end{pmatrix}
\place{\small $k_1$}{-290mu}{18pt}
\place{\small $k_2$}{-290mu}{6pt}
\place{\small $k_3$}{-290mu}{-6pt}
\place{\small $k_4$}{-290mu}{-18pt}
\place{\small $k_1$}{-220mu}{30pt}
\place{\small $k_2$}{-160mu}{30pt}
\place{\small $k_3$}{-100mu}{30pt}
\place{\small $k_4$}{-40mu}{30pt}
 \enspace .
\end{equation*}
Indeed, vertices $(k_\alpha, \sigma)$ and $(k_\alpha, \sigma')$ are not shared by any paired edges~$j$, i.e., $\braket{(k_\alpha,\sigma)}{v_j} \braket{v_j}{(k_\alpha, \sigma')} = 0$, unless $\sigma = \sigma'$.  If $\sigma = \sigma'$, then each of~$n$ paired edges contributes $\frac{1}{8 n}$.  A similar argument holds for $\Delta(k_1, k_4)$ and $\Delta(k_2, k_3)$, except in these cases the paired edges each contribute $-\frac{1}{8 n}$.  For $\alpha \in \{1, 4\}$ and $\beta \in \{2, 3\}$, any two vertices $(k_\alpha, \sigma)$, $(k_\beta, \sigma')$ are shared by exactly one paired edge.  
\end{enumerate}
Observe that~$\Delta$ is a constant-sized block matrix, where each block is the sum of a constant multiple of~$I_n$ and a constant multiple of~$J_n/n$.  By \refclm{blockmatrixspectrum}, the set of eigenvalues of $\Delta$ does not depend on~$n$.  In particular, it has an $\Omega(1)$ spectral gap from zero, as desired.  

We now show that both $R_A$ and $R_B$ can be implemented efficiently.  As described earlier, the algorithm works in the Hilbert space spanned by vectors $\ket i \otimes \ket j$, where $j$ varies over input vectors and the target vector, and $i$ varies over vertices for which $\braket{i}{v_j} \neq 0$.  A more convenient representation for such pairs $(i, j)$ is as a tuple $(k, \sigma, \tau, s) \in [\ell] \times [n] \times [n] \times \{+, -\}$, where~$k$ specifies one of the $\ell$ layers, and $\sigma, \tau$ specify the endpoints of an edge from layer~$k$ either to the next layer (if $s = +$) or to the previous layer (if $s = -$).  This representation works whether~$j$ is an ordinary edge or a pair of paired edges.  Two additional tuples, $(1,1,0,-)$ and $(\ell,1,0,+)$, are needed, though, for~$j$ the second edge from~$s$ to~$t$, i.e., the never-available input vector $\sqrt{1 - 1/\alpha^2} (\ket t - \ket s)$.  The states $\ket{k, \sigma, \tau, s}$ can be stored using a logarithmic number of qubits.  

We start with the description of $R_A$.  For all $i = (k, \sigma) \in [\ell] \times [n]$ except $s = (1, 1)$ and $t = (\ell, 1)$, $\ket i \otimes \ket{a_i}$ is the uniform superposition of the states $\{\ket{k, \sigma, \tau, s} : \tau \in [n], s \in \{+, -\} \}$, so the reflection is a Grover diffusion operation.  For $(k, \sigma) = (1, 1) = s$, we perform a slightly different operation.  Let $F$ be the Fourier transform on the space spanned by $\{\ket{1,1,\tau,\pm} : \tau \in [n]\}$ that maps $\ket{1,1,1,-}$ to the uniform superposition; let $K$ be a unitary on $\spn \{\ket{1,1,1,-}, \ket{1,1,0,-}\}$ that maps $\ket{1,1,1,-}$ to $(1/\alpha) \ket{1,1,1,-} + \sqrt{1-1/\alpha^2} \ket{1,1,0,-}$; and let $L$ multiply every phase, except that of $|1,1,1,-\rangle$, by~$-1$.  Then the necessary transformation can be implemented as $F K L K^{-1} F^{-1}$.  A similar operation works for $(k, \sigma) = (\ell, 1) = t$.  

The implementation of $R_B$ is similar.  For layers $k$ and $k+1$ with only ordinary edges between them, it suffices to apply the negated swap to all pairs $(\ket{k,\sigma,\tau,+}, \ket{k+1,\tau,\sigma,-})$.  This can be done in logarithmic time.  For paired layers, the reflection about $\ketbra{b_j}{b_j}$ is performed in a four-dimensional subspace.  

Finally, for the implementation of $R_\Pi$ we need to clarify the use of the random coloring.  One solution is to generate random numbers classically, and provide them in the form of an oracle mapping $\sigma \in [n]$ to the color of vertex~$\sigma$.  This requires coherent access to $\Theta(n)$-bit string.  For \refthm{star}, however, one can reduce the space complexity to $O(\log n)$, by using a $C$-uniform hash function family from~$[n]$ to~$[C]$, where $C$ is the total number of colors.  If necessary, we may assume that $n$ and~$C$ are powers of two.  $C$-wise independence is enough for the proof.  For \refthm{triangle}, this does not work, though, because we need to ensure that the negative witness size is small with high probability.  

Consider layer $k$ that corresponds to color~$c$.  A vertex $(k, \sigma)$ corresponds to vertex~$\sigma$ of $G$ if and only if it has color~$c$.  Otherwise, it is a dummy vertex.  To check whether the edge is available, the algorithm checks the layers that it connects.  If they are connected by ordinary edges of~$H$, it queries both endpoints to check if they have the correct colors.  If they do, it executes the input oracle, to check for the availability of the edge.  If the edge is available, it does nothing.  In all other cases, it negates the phase of the state.  
\pfend

\section{Algorithm for detecting a star with two subdivided edges} \label{sec:startwo}

In the introduction, via a reduction to $st$-connectivity, we gave an $O(n)$-query quantum algorithm for detecting the presence of a fixed-length path in an $n$-vertex graph~$G$ given by its adjacency matrix.  In this section, we generalize the reduction to the problem of detecting as a subgraph a star with two subdivided legs.  

\begin{thm} \label{thm:startwo}
Fix $T$ a star with two subdivided legs.  Then there exists a quantum algorithm that, given query access to the adjacency matrix of a simple graph $G$ with $n$ vertices, makes $O(n)$ queries, and except with error probability at most~$1/3$ accepts if and only if $G$ contains $T$ as a subgraph.  
\end{thm}

\begin{proof}
Let the vertex set of $T$ be $V_T = \{ v_1, \ldots, v_k \} \cup \{ w_1, \ldots, w_d \}$, where the edges are $\{ (v_j, v_{j+1}) : j \in [k-1] \} \cup \{ (v_\ell, w_i) : i \in [d] \}$; the vertex $v_\ell$ is the hub of the star.  Without loss of generality, we may assume that $\ell \in \{2, \ldots, k-1\}$.  Let $c$ be a uniformly random map from the vertex set of~$G$ to~$V_T$.  

Define an instance of $st$-connectivity by transforming~$G$ into a graph~$H$ as follows.  The vertex set of~$H$ is 
\begin{equation*}
\{ s, t \} \cup c^{-1}(\{v_1, \ldots, v_{\ell-1}, v_{\ell+1}, \ldots, v_k\}) \cup 
\!\!\bigcup_{\substack{u \in c^{-1}(v_\ell)}}\!\! \Big( \{ u_0, \ldots, u_d \} \cup \big\{ \mu_u : \mu \in c^{-1}(\{ w_1, \ldots, w_d \}) \big\} \Big)
 \enspace .
\end{equation*}
Thus each vertex $u \in c^{-1}(v_\ell)$ is split into $1 + d$ copies, whereas each vertex $\mu \in c^{-1}(w_i)$ is split into $\abs{c^{-1}(v_\ell)}$ copies.  There are $O(n^2)$ vertices total.  

There are free edges $(s, u)$ for every $u \in c^{-1}(v_1)$, and $(u, t)$ for every $u \in c^{-1}(v_k)$.  
For every $u \in c^{-1}(v_j)$ and $u' \in c^{-1}(v_{j+1})$, if $G$ has the edge $(u, u')$, then $H$ has the edge either $(u, u')$ if $\ell \notin \{j, j+1\}$, $(u, u'_0)$ if $\ell = j+1$, or $(u_d, u')$ if $\ell = j$.  Finally, for every edge $(u, \mu) \in c^{-1}(v_\ell) \times c^{-1}(w_i)$, $H$ has the two edges $(u_{i-1}, \mu_u)$ and $(\mu_u, u_i)$.  The total number of possible edges, i.e., input vectors for the $st$-connectivity span program, is $\abs{c^{-1}(\{v_1, v_k\})} + \prod_{j \in [k-1]} \abs{c^{-1}(v_j)} \abs{c^{-1}(v_{j+1})} + 2 \abs{c^{-1}(v_\ell)} \abs{c^{-1}(\{w_i\})} = O(n^2)$.  An example is given in \reffig{startwo}.  

\begin{figure}
\vspace{-.4cm}
\centering
\subfigure[]{\raisebox{.74cm}{\includegraphics{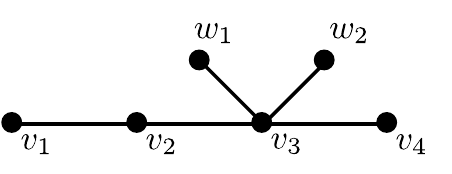}}}
$\quad$
\subfigure[]{\raisebox{0cm}{\includegraphics{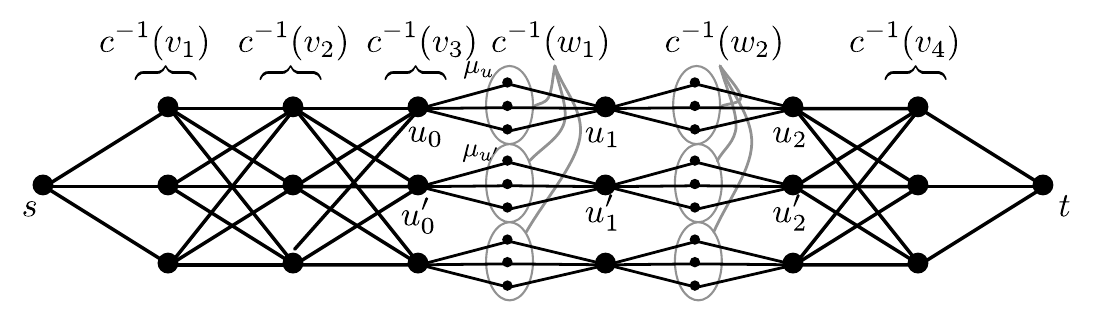}}}
\caption{(a) $T$ a star with one subdivided edge; $k = 4$, $\ell = 3$, $d = 2$.  (b) The vertices and set of possible edges in~$H$ for an input with $n = 18$ vertices colored evenly, i.e., $\abs{c^{-1}(v)} = 3$ for all $v \in V_T$.}
\label{fig:startwo}
\end{figure}

If $G$ contains~$T$ as a subgraph, there is a positive constant probability that it is colored correctly.  Then $s$ is connected to~$t$ by a path of length $k+2d+1$, and the span program witness size is $O(1)$.  

If $G$ does not contain~$T$ as a subgraph, then the construction guarantees that $s$ is not connected to~$t$.  Indeed, for there to be a path $u_0, \mu_u, u_1, \ldots, u_d$ through~$H$, $G$ must have a $d$-leg star centered at~$u$, and for there to be paths from $s$ to~$u_0$ and from~$u_d$ to~$t$, $G$ must further have vertex-disjoint paths of lengths at least~$\ell - 1$ and~$k - \ell$ attached to~$u$.  The witness size is $O(n^2)$, since there are only $O(n^2)$ input vectors.  
\end{proof}

Unlike in \refthm{star}, the breadcrumb trick of using paired edges is not needed here.  The path through~$G$ induced by a path through~$H$ from~$u_0$ to~$u_d$ goes at most one step from the hub vertex $u \in c^{-1}(v_\ell)$.  Instead of using paired edges to remember where the path came from, it is therefore enough to make copies of the vertices in $c^{-1}(\{w_1, \ldots, w_d\})$.  Unlike in \refthm{star}, no promise on the input~$G$ is required for \refthm{startwo}.  Observe, however, that the derived graph~$H$ does have a certain promised structure, which ensures that there are only $O(n^2)$ possible edges or span program input vectors, even though $H$ has $O(n^2)$ vertices.  

We omit the details, but it can be shown along the lines of \refthm{efficient} that this algorithm can be implemented in logarithmic space, with only a poly-logarithmic time overhead.  

The same technique as used in \refthm{startwo}, except splitting up every vertex in $c^{-1}(\{v_1, \ldots, v_k\})$, works to give an $O(n)$-query quantum algorithm for the subgraph/not-a-minor promise problem for a fixed ``fuzzy caterpillar" graph~$T$, having vertices $\{ v_1, \ldots, v_k \} \cup \bigcup_{j \in [k]} \{ w_{j, 1}, \ldots, w_{j, d_j} \}$, and edges $\{ (v_j, v_{j+1}) : j \in [k-1] \} \cup \{ (v_j, w_{j, i}) : j \in [k], \, i \in [d_j] \}$.  The algorithm does not solve the subgraph-detection problem for this~$T$ because the $s$-$t$ path can zig-zag back and forth between vertices in $c^{-1}(v_j)$ and $c^{-1}(v_{j+1})$.

\paragraph{Acknowledgments.}  

Much of this research was conducted at the Institute for Quantum Computing, University of Waterloo. A.B.~would like to thank Andrew Childs for hospitality; Andrew Childs, Robin Kothari, and Rajat Mittal for many useful discussions; and Jim Geelen for sharing the example in \reffig{skew}.  

A.B.~has been supported by the European Social Fund within the project ``Support for Doctoral Studies at University of Latvia.''  B.R.~acknowledges support from NSERC, ARO-DTO and Mitacs.

\bibliographystyle{alpha-eprint}
\bibliography{q}

\end{document}